\documentclass[12pt]{article}
\usepackage{amsfonts}
\usepackage{amsmath}

\usepackage[english]{babel}
\usepackage[T1]{fontenc}

\usepackage[top=30truemm,bottom=30truemm,left=25truemm,right=25truemm]{geometry}

\newtheorem{theorem}{Theorem}
\newtheorem{definition}{Definition}
\newtheorem{lemma}{Lemma}

\newtheorem{proposition}{Proposition}

\newtheorem{proof}{Proof}

\def\<{\langle}
\def\>{\rangle}

\begin{document}

\centerline{\bf A Formulation of R\'enyi Entropy on $C^*$-Algebras}

\bigskip\bigskip

\centerline{\sc Farrukh Mukhamedov}
\vspace{+1mm}
\centerline{\it  Department of Mathematical Sciences,}
\centerline{\it United Arab Emirates University,}
\centerline{\it 15551 Al-Ain, United Arab Emirates}
\centerline{E-mail: {\tt far75m@yandex.ru,\quad farrukh.m@uaeu.ac.ae}}
\bigskip
\centerline{\sc Kyouhei Ohmura}
\vspace{+1mm}
\centerline{\it Department of Information Sciences,}
\centerline{\it Tokyo University of Science,}
\centerline{\it Noda City, Chiba 278-8510, Japan}
\centerline{E-mail: {\tt 6317701@ed.tus.ac.jp, \quad ohmura.kyouhei@gmail.com}}
\bigskip
\centerline{\sc Noboru Watanabe}
\vspace{+1mm}
\centerline{\it Department of Information Sciences,}
\centerline{\it Tokyo University of Science,}
\centerline{\it Noda City, Chiba 278-8510, Japan}
\centerline{E-mail: {\tt watanabe@is.noda.tus.ac.jp}}

\bigskip\bigskip

\centerline{\bf Abstract }

\bigskip%\bigskip
\noindent
The entropy of probability distribution defined by Shannon has several extensions. R\'enyi entropy is  one of the general extensions of Shannon entropy and is widely used in engineering, physics, and so on. On the other hand, the quantum analogue of Shannon entropy is von Neumann entropy. Furthermore, the formulation of this entropy was extended to on $C^*$-algebras by Ohya ($\mathcal{S}$-mixing entropy). In this paper, we formulate Renyi entropy on $C^*$-algebras based on $\mathcal{S}$-mixing entropy and prove several inequalities for the uncertainties of states in various reference systems.


\tableofcontents

\section{Introduction}
\label{intro}
Shannon introduced the entropy as the information amount of information systems represented by probability spaces \cite{shannon}. R\'enyi defined a general extension of Shannon entropy on probability spaces which is called R\'enyi entropy \cite{renyi}. R\'enyi entropy is more general than Shannon entropy in the sense of a positive number $\alpha$, and it corresponds to Shannon entropy when $\alpha \to 1$. This entropy is useful and widely used in physics, engineering, and so on \cite{appl}, \cite{taka}.

On the other hand, von Neumann entropy measures the complexity (or the information amount) of a quantum system \cite{vn}. In 1984, Ohya formulated the general extension of von Neumann entropy which is called $\mathcal{S}$-mixing entropy on $C^*$-algebras \cite{smix},\cite{itsuse}, \cite{makino}, \cite{noteon}. $\mathcal{S}$-mixing  entropy depends on choosing subset (reference system) of the set of all states on the $C^*$-algebra. Thanks to the property, one can measures the uncertainty of the state depending on reference systems. Mukhamedov and Watanabe formulated an extension of $\mathcal{S}$-mixing entropy by taking the set of all quantum channels as the reference system. Moreover, they showed that the entropy can apply to detect entangled states and calculated the complexities of qubit and phase-damping channels \cite{mw}. 

In this paper, we formulate R\'enyi entropy on $C^*$-algebras based on $\mathcal{S}$-mixing entropy and show that the introduced entropy corresponds to $\mathcal{S}$-mixing entropy when $\alpha \to 1$. Furthermore, we prove  that our R\'enyi entropy is a general extension of quantum R\'enyi entropy \cite{petz}, \cite{qent} if $\alpha > 1$. Moreover, by using our R\'enyi entropy, we investigate the uncertainties of states measured from various reference systems.

We organize the paper as follows: In Section 2, we recall the notations and some properties of the R\'enyi entropy on probability spaces. Furthermore, we review the decomposition theory of states on $C^*$-algebras and the definition of $\mathcal{S}$-mixing  entropy. In Section 3, we formulate R\'enyi entropy on $C^*$-algebras based on the definition of $\mathcal{S}$-mixing entropy and show several properties of it. Furthermore, by using the introduced entropy, we prove the equalities or inequalities of the complexities of  states measured from different reference systems.

%Preliminaries------

\section{Preliminaries}
\label{sec:1}
In this section, we review the definitions of R\'enyi entropy and $\mathcal{S}$-mixing entropy, and those several properties. 

\subsection{R\'enyi Entropy}
\label{sec:2}
In this chapter, $\log$ denotes the logarithm of base $2$.

%Ddef of Renyi
\begin{definition}
Let $\{ p_1, p_2, \cdots , p_n \}$ be the  probability distribution of a random variable $X$. The {\it R\'enyi entropy} is defined by
\begin{equation}
S_{\alpha} (X) := \frac{1}{1 - \alpha} \log \sum_{k=1}^n p_k^{\alpha}\quad , \quad \alpha \in [0, +\infty)\backslash \{ 1\}.
\end{equation}
\end{definition}
This entropy corresponds to the Shannon entropy when $\alpha \to 1$. Namely, the following theorem holds.

%Theorem of Renyi = Shannon
\begin{theorem}
Under the above assumptioms,
\begin{equation}
\lim_{\alpha \to 1} S_{\alpha} (X) = - \sum_{k=1}^n p_k \log p_k
\end{equation}
is satisfied.
\end{theorem}

%Additivity of Renyi
\noindent Furthermore, R\'enyi entropy has the additivity.
\begin{theorem}
If $X$ and $Y$ are independent random variables,
\begin{equation}
S_{\alpha} (X, Y) = S_{\alpha} (X) + S_{\alpha} (Y).
\end{equation}
\end{theorem}

%Decreasing
\noindent Moreover, since
$$
\frac{\partial}{\partial \alpha} S_{\alpha} \le 0,
$$
one can see that this entropy is a decreasing function with respect to the parameter $\alpha$.\\

%Coding
\noindent R\'enyi entropy has important roles for the coding theory. For instance, the following theorem exists for the entropy \cite{camp}, \cite{petz}. \\
\noindent Let $\mathcal{X}$ be a finite alphabet set and $X$ be a rondam variable of $\mathcal{X}$. Let $C$ be a source code, that is, a map from $\mathcal{X}$ to the set of finite-length strings of symbols of a binary alphabet. Then $C (x)$ denotes the codeword of $x \in \mathcal{X}$ and $l(x)$ denotes the length of $C (x)$. Now we define the cost of the coding:
$$
L_{\beta} (C) := \frac{1}{\beta} \log \sum_x p (x) 2^{\beta l (x)}
$$
where $p(x)$ is the pbability of $x$ and $\beta > -1$.

%Coding theorem
\begin{theorem}
Let $\alpha = 1/ (1+\beta)$. For a uniquely decodable code, the following inequality holds:
\begin{equation}
L_{\beta} (C) \geq S_{\alpha} (X).
\end{equation}
Furthermore, there exists a uniquely decodable code $C$ satisfying
\begin{equation}
L_{\beta} (C) \le S_{\alpha} (X) + 1.
\end{equation}
\end{theorem}

%Dcomposition Theory-----------

\subsection{Decomposition Theory}
A quantum state can be decomposed into simpler components. In this section, we recall the mathematical theory on the decompositions of states \cite{br1}, \cite{qent} that we need as follows.\\

\noindent Let $(\mathcal{A}, \mathfrak{S}, \theta (G))$ be a $C^*$-dynamical system, that is, $\mathcal{A}$ is a $C^*$-algebra, $\mathfrak{S}$ is the set of all states $\varphi$  on $\mathcal{A}$, and $\theta (G)$ is the set of all  *-automorphisms on $\mathcal{A}$ associated with a group $G$. The triplet $(\mathcal{A}, \mathfrak{S}, \theta (G))$ describes the dynamics of a quantum system \cite{qent}. \\
\noindent Moreover, let $I (\theta)$ be the set of all $\theta$-invariant states (i.e. $\varphi \circ \theta_g = \varphi\ ,\ \forall g \in G$), and  $K_{\beta} (\theta)\ (G = \mathbb{R})$ be the set of all states satisfying KMS condition with respect to $\theta_t$ ($t \in \mathbb{R}$). 

%Def of ergodic decomposition
\begin{definition}
The decomposition from an $\theta$-invariant state into extremal $\theta$-invariant states is called {\it ergodic decomposition}.
\end{definition}

\noindent Since $I (\theta)$ and $K_{\beta} (\theta)$ are weak*-compact and convex subset of $\mathfrak{S}$, 
we deal with the case where spaces have such conditions.

\noindent Let $\mathcal{S}$ be a compact and convex subspace  of a locally convex Hausdorff space. Moreover, let ${\rm ex}\mathcal{S}$ be the set of all extreme points of $\mathcal{S}$. According to the Krein-Mil'man theorem \cite{krein}, ${\rm ex}\mathcal{S} \neq \phi$ and the weak*-closure of convex hull of ${\rm ex}\mathcal{S}$ equals to $\mathcal{S}$, i.e. $\overline{co}^{w^*} {\rm ex}\mathcal{S} = \mathcal{S}$.

%Def of extremal
\begin{definition}
The decomposition from $\mathcal{S}$ into ${\rm ex}\mathcal{S}$ is called {\it extremal decomposition}.
\end{definition}

%M (S)
\noindent Let $M(\mathcal{S})$ be the set of all normal Borel measures on $\mathcal{S}$. Furthermore, define
\begin{equation}
M_1 (\mathcal{S}) := \{ \mu \in M (\mathcal{S}),\ \mu (\mathcal{S}) = 1 \}.
\end{equation}

%Def of barycenter
\begin{definition}
For any $\mu \in M (\mathcal{S})$,
\begin{equation}
b (\mu ) := \int_{\mathcal{S}} \omega d \mu (\omega) 
\end{equation}
is called the {\it barycenter} of $\mu$.
\end{definition}

\noindent Moreover, let $C_{\mathbb{R}}(\mathcal{S})$ be the set of all real continuous functions on $\mathcal{S}$ and
%K(S)
$$
K(\mathcal{S}) := \{ f \in C_{\mathbb{R}} (\mathcal{S})\ ;\ f\ {\rm are \ convex\ functions}\}.
$$
\noindent For two measures $\mu , \nu \in M (\mathcal{S})$, define ``$\prec$'' as follows :
$$
\mu \prec \nu \overset{def}{\iff} \mu (f) \le \nu (f),\quad \forall f \in K(\mathcal{S}).
$$
Then $\prec$ gives an ordering on $M(\mathcal{S})$. Let us denote $M^m (\mathcal{S})$ as the set of all maximal elements with respect to the ordering.\\

\noindent Furthermore, we recall the following theorems.

%Th of S1
\begin{theorem}
If $\mathcal{S}$ is a metricable compact convex set ;
\begin{enumerate}
\item ${\rm ex} \mathcal{S}$ is a $G_{\delta}$ set.
\item $\mu \in M_1^m (\mathcal{S})$ iff $\mu ({\rm ex} \mathcal{S}) = 1$.
\item For any $\varphi \in \mathcal{S}$, there exist $\mu \in M_1^m (\mathcal{S})$ such that $\varphi = b(\mu)$.
\end{enumerate}
\end{theorem}

%Th of S2
\begin{theorem}
If $\mathcal{S}$ is a compact convex set ;
\begin{enumerate}
\item Any $\mu \in M_1^m (\mathcal{S})$ has ${\rm ex}\mathcal{S}$ as their pseudo-support (i.e. for any Bair sets $Q$ such that ${\rm ex}\mathcal{S} \subset Q \subset \mathcal{S}$, $\mu (Q) = 1$).
\item For any $\varphi \in \mathcal{S}$, there exist $\mu$ which satisfy (1) such that $\varphi = b (\mu)$. 
\end{enumerate}
\end{theorem}

\noindent Moreover, we have the following theorem for uniqueness of maximal measure $\mu$.\\

\noindent Let $\mathcal{X}$ be a locally convex Hausdorff space, $\mathcal{S}$ be a compact convex subset of $\mathcal{X}$, and $\mathcal{K}$ be a convex cone whose vortex is 0.  Furthermore, let $\mathcal{S}$ be the base of $\mathcal{K}$, i.e.
$$
\mathcal{K} = \{ \lambda \omega \ ;\ \lambda \geq 0,\ \omega \in \mathcal{S} \}.
$$
Then $\mathcal{K}$ is the convex cone generated by $\{ 1\} \times \mathcal{S}$. Defining
$$
\omega_1 \geq \omega_2 \overset{def}{\iff} \omega_1 - \omega_2 \in \mathcal{K},
$$
then $\geq$ gives an ordering on $\mathcal{K}$.

%Def of Shoquet 
\begin{definition}
If $\mathcal{K}$ is the lattice with respect to the above $\geq$, $\mathcal{S}$ is called {\it Choquet simplex}.
\end{definition}

%Th S Choquet Simplex
\begin{theorem}\label{ifS}
If $\mathcal{S}$ is compact convex, the following are equivalent:
\begin{enumerate}
\item $\mathcal{S}$ is a Choquet simplex.
\item For any $\varphi \in \mathcal{S}$, there exists a unique maximal probability measure $\mu$.
\end{enumerate}
\end{theorem}

\noindent Let $M_{\varphi} (\mathcal{S})$ be the set of all $\mu$ which is its barycenter equals to the state $\varphi$ on the $C^*$-algebra, i.e.
\begin{equation}\label{Mphi}
M_{\varphi} (\mathcal{S}) := \{ \mu \in M_1 (\mathcal{S}),\ b (\mu) = \varphi \}.
\end{equation}

%Barycebntric Dec
\noindent For $\varphi$ satisfying (\ref{Mphi}), one obtains the integral representation of $\varphi$:
\begin{equation}\label{bcdec}
\varphi =  \int_{\mathcal{S}} \omega d \mu (\omega).
\end{equation}
It is called the {\it barycentric decomposition} of $\varphi$. According to Theorem \ref{ifS}, this dcomposition is not unique unless $\mathcal{S}$ is a Choquet simplex.\\

%Orthogonality
\noindent Furthermore, we review the orthogonality of states. Let $\{ \mathcal{H}_{\varphi}, \pi_{\varphi}, x_{\varphi} \}$ be the GNS representation defined by $\varphi$. For $\varphi_1 ,\ \varphi_2 \in \mathfrak{S}$, set $\varphi := \varphi_1 + \varphi_2 \in \mathcal{A}_+^*$. Then the following are euivalent:
 \begin{enumerate}
\item Let $\psi \in \mathcal{A}_+^*$. If $\psi \le \varphi_1$ and  $\psi \le \varphi_2$, $\psi = 0$.
\item There exists a projection $E \in \pi_{\varphi} (\mathcal{A})'$ such that
\begin{eqnarray*}
\varphi_1 (A) &=& \< x_{\varphi}, E \pi_{\varphi} (A) x_{\varphi} \>,\\
\varphi_2 (A) &=& \< x_{\varphi}, (I - E) \pi_{\varphi} (A) x_{\varphi} \>.
\end{eqnarray*}
\item $\mathcal{H}_{\varphi} = \mathcal{H}_{\varphi_1} \oplus \mathcal{H}_{\varphi_2}$, $\pi_{\varphi} = \pi_{\varphi_1} \oplus \pi_{\varphi_2}$, $x_{\varphi} = x_{\varphi_1} \oplus x_{\varphi_2}$.
\end{enumerate}

%Def of orthogonality
\begin{definition}
The states $\varphi_1$, $\varphi_2$ satisfying the above conditions are called {\it mutually orthogonal} and denoted by $\varphi_1 \perp \varphi_2$.
\end{definition}

%Def of orthogonal measure
\begin{definition}
For any Borel sets $Q \subset \mathfrak{S}$ (i.e. $Q \in \mathcal{B}(\mathfrak{S})$), $\mu \in M (\mathfrak{S})$ satisfying
$$
\left( \int_Q \omega d \mu \right) \perp \left( \int_{\mathfrak{S}\backslash Q} \omega d \mu \right)
$$
is called {\it orthogonal measure} on $\mathfrak{S}$. 
\end{definition}

\noindent We define $\mathcal{O}_{\varphi} (\mathfrak{S})$ as the set of all orthogonal probability measures whose barycenters are $\varphi$.

%S-Miixing Entropy--------------------

\subsection{$\mathcal{S}$-Mixing Entropy}
If $\mu \in M_{\varphi} (\mathcal{S})$ has countable supports, that is, (\ref{bcdec}) can be written as 
\begin{equation}\label{count}
\varphi = \sum \lambda_k \varphi_k
\end{equation} 
where $\lambda_k > 0$ ; $\sum \lambda_k = 1$ and $\{ \varphi_k \} \subset {\rm ex} \mathcal{S}$, we denote the set of all such measures as $D_{\varphi} (\mathcal{S})$.

%Def of S mixing entropy
\begin{definition}\label{Smix}
Under the above assumptions, the entropy of $\varphi \in \mathcal{S}$ is given by
\begin{equation}
S^{\mathcal{S}} (\varphi ) := 
\begin{cases}
\inf \{ - \sum \lambda_k \log \lambda_k ;\ \mu = \{ \lambda_k \} \in D_{\varphi} (\mathcal{S}) \} \\
+ \infty \qquad (\mu \notin D_{\varphi} (\mathcal{S}))
\end{cases}
\end{equation}
\end{definition}

\noindent The above entropy is called {\it $\mathcal{S}$-mixing entropy}. Since  one can regard that the complexity of the system is $+ \infty$ if $\varphi$ has uncountable states, Ohya defined $S^{\mathcal{S}} (\varphi ) := + \infty \  (\mu \notin D_{\varphi} (\mathcal{S}))$.\\
\noindent $S^{\mathcal{S}} (\varphi )$ depends on the set $\mathcal{S}$ chosen, thus it represents the amount of complexity of the state measured from the reference system $\mathcal{S}$. That is, this entropy takes measuring the uncertainty of states from various reference systems into account. \\
\noindent Furthermore,  if $\varphi$ is faithful normal and $\mathcal{S} = \mathfrak{S}$, this entropy corresponds to von Neumann entropy \cite{smix}, \cite{qent}.\\

%Sreal
\noindent By the way, since one can regard that the complexities of  real physical systems are finite, we denote the subset of $\mathcal{S}$ as
$$
\mathcal{S}_r := \{ \varphi \in \mathcal{S}\ ;\ S^{\mathcal{S}} (\varphi) < \infty\}.
$$
Since $\mathcal{S} = \overline{co}^{w^*}{\rm ex} \mathcal{S}$, the following proposition holds.
\begin{proposition}
\begin{equation}
\bar{\mathcal{S}}^{w^*}_r = \mathcal{S}.
\end{equation}
\end{proposition}

%Renyi Entropy on The C^* algebra---------

\section{R\'enyi Entropy on $C^*$-Algebras}
In this section, we define R\'enyi entropy on $C^*$-algebras based on $\mathcal{S}$-mixing entropy and show that the introduced entropy includes $\mathcal{S}$-mixing entropy and quantum R\'enyi entropy as the special cases. Furthermore, by using our R\'enyi entropy, we investigate the uncertainty of states in different reference systems.

%Def of Renyi entropy
\begin{definition}
Under the same assumptions and notations with Definition \ref{Smix}, we define:
\begin{equation}\label{renyi}
S_{\alpha}^{\mathcal{S}} (\varphi) := \inf \left\{ (1 - \alpha )^{-1} \log \sum_k \lambda_k^{\alpha} \right\}
\quad ; \quad
\alpha \in [0, +\infty) \backslash \{ 1\}
\end{equation}
where the infimum is taken over all $\mu = \{ \lambda_k \} \in D_{\varphi} (\mathcal{S})$. Moreover, if $\mu \notin  D_{\varphi} (\mathcal{S})$, $S_{\alpha}^{\mathcal{S}} (\varphi) := \infty$.
\end{definition}
We call (\ref{renyi}) {\it $\mathcal{S}$-mixing R\'enyi entropy}.\\
From the analogue of classical case, one can see the following theorem:

%Th of decreasing 
\begin{theorem}
$S_{\alpha}^{\mathcal{S}} (\varphi)$ is monotone decreasing with respect to the parameter $\alpha$.
\end{theorem}

\noindent Furthermore, in analogy with the classical case, we have the following theorem.

%Theorem Renyi and S-mix
\begin{theorem}
For any $\varphi \in \mathcal{S}$,
\begin{equation}\label{th}
\lim_{\alpha \to 1} S_{\alpha}^{\mathcal{S}} (\varphi) = S^{\mathcal{S}} (\varphi)
\end{equation}
holds.
\end{theorem}
%proof
\begin{proof}
According to the classical case, for $\mu \in D_{\varphi} (\mathcal{S})$,
\begin{equation}\label{clanalogue}
\lim_{\alpha \to 1} (1 - \alpha)^{-1} \log \sum_k \lambda_k^{\alpha} = - \sum_k \lambda_k \log \lambda_k
\end{equation}
holds. We shall denote $\tilde{S}_{\alpha}^{\mathcal{S}} (\varphi) := (1 - \alpha)^{-1} \log \sum_k \lambda_k^{\alpha}$, $\tilde{S}^{\mathcal{S}} (\varphi) :=  - \sum_k \lambda_k \log \lambda_k$. 
%From the classical case, for any $\alpha > 1$, $\tilde{S}^{\mathcal{S}} (\varphi) > \tilde{S}_{\alpha}^{\mathcal{S}} (\varphi)$ is satisfied. 
Then we have
$$
0 \le \inf_{\{\lambda_k\}} \tilde{S}^{\mathcal{S}} (\varphi) - \inf_{\{ \lambda'_k \}} \tilde{S}_{\alpha}^{\mathcal{S}} (\varphi) = \sup (- \tilde{S}_{\alpha}^{\mathcal{S}} (\varphi)) - \sup (-\tilde{S}^{\mathcal{S}} (\varphi))
$$
\begin{equation}\label{le}
\le \sup ( \tilde{S}^{\mathcal{S}} (\varphi) - \tilde{S}_{\alpha}^{\mathcal{S}} (\varphi)) \quad ,\quad \forall \alpha > 1.
\end{equation}
\begin{equation}\label{le2}
0 \le  \inf_{\{ \lambda'_k \}} \tilde{S}_{\alpha}^{\mathcal{S}} (\varphi) -  \inf_{\{\lambda_k\}} \tilde{S}^{\mathcal{S}} (\varphi) \le \sup (\tilde{S}_{\alpha}^{\mathcal{S}} (\varphi) -  \tilde{S}^{\mathcal{S}} (\varphi)) \quad ,\quad 0\le \forall \alpha < 1.
\end{equation}
Due to (\ref{clanalogue}), the right hand sides of (\ref{le}) and (\ref{le2}) go to $0$ when $\alpha \to 1$. Therefore we obtain the theorem.

\end{proof}

\noindent Now we prove that our $\mathcal{S}$-mixing R\'enyi entropy includes the density case \cite{petz}, \cite{qent}. Let $\mathbf{T} (\mathcal{H})$ be the set of all trace class operators on a Hilbert space $\mathcal{H}$, and $\mathbf{T} (\mathcal{H})_{+,1} := \{ A \in \mathbf{T}(\mathcal{H})\ ;\ \mathrm{Tr}A = 1\}$.

%Def Renyi of density
\begin{definition}
For any $\rho \in \mathbf{T}(\mathcal{H})_{+,1}$ and any $\alpha \in [0, +\infty) \backslash \{ 1 \}$, the {\it quantum R\'enyi entropy} is defined by
\begin{equation}
S_{\alpha} (\rho) := (1 - \alpha )^{-1} \log {\rm Tr} \rho^{\alpha}.
\end{equation}
\end{definition}

%Lemma of inf

\begin{lemma}\label{lem1}
Let $\rho = \sum_n \lambda_n \rho_n$ be the decomposition into pure states (i.e. $\dim (\mathrm{ran} \rho_n) = 1$). For any $\alpha > 1$,
\begin{equation}\label{lem}
S_{\alpha} (\rho) \le (1 - \alpha)^{-1} \log \sum_n \lambda_n^{\alpha}
\end{equation}
holds. If $\rho_n \perp \rho_m \ (n \neq m)$, one obtains the equality.
\end{lemma}
%proof
\begin{proof}
Let $\rho = \sum_k p_k E_k$ be the Schatten decomposition \cite{sch} of $\rho$. Then for any $n \in \mathbb{N}$,
$$
\sum_{k=1}^n p_k \geq \sum_{k=1}^n \lambda_n
$$
is satisfied \cite{qent}. Therefore we have $\sum_{k=1}^n p_k^{\alpha} \geq \sum_{k=1}^n \lambda_n^{\alpha}$ ($\forall \alpha \in [0, +\infty) \backslash \{ 1\}$). Moreover, according to the monotonicity of $\log$,
$$
(1 - \alpha)^{-1} \log \sum_{k=1}^n p_k^{\alpha} \le (1 - \alpha)^{-1} \log \sum_{k=1}^n \lambda_k^{\alpha} \quad , \quad  \forall \alpha > 1.
$$
Since $0 \le \sum_{k=1}^n p_k^{\alpha} < 1$ (resp. $0 \le \sum_{k=1}^n \lambda_n^{\alpha} < 1$), there exists the limit : $\displaystyle \lim_{n \to \infty} \log \sum_{k=1}^n p_k^{\alpha}$ (resp. $\displaystyle \lim_{n \to \infty} \log \sum_{k=1}^n \lambda_k^{\alpha}$). Thus we have
$$
(1 - \alpha)^{-1} \log \sum_{k=1}^{\infty} p_k^{\alpha} \le (1 - \alpha)^{-1} \log \sum_{k=1}^{\infty} \lambda_k^{\alpha} \quad , \quad  \forall \alpha > 1.
$$
This gives the inequality (\ref{lem}). \\
Moreover, if $\rho_n \perp \rho_m\ (n \neq m)$, $\rho = \sum_n \lambda_n \rho_n$ becomes the Schatten decomposition of $\rho$. Thus $\lambda_n = p_n$. Therefore
$$
S_{\alpha} (\rho) = (1 - \alpha)^{-1} \log \sum_n \lambda_n^{\alpha}.
$$
\end{proof}

\noindent Using this lemma, we prove the following theorem.

%Th Renyi of \rho
\begin{theorem}
 Let $\mathcal{A}$ be a $C^*$-algebra. If a state $\varphi$ can be written as $\varphi (A) = \mathrm{Tr} \rho A\ (\forall A \in \mathcal{A})$, 
\begin{equation}
S_{\alpha}^{\mathfrak{S}} (\varphi) = S_{\alpha} (\rho) \quad , \quad \forall \alpha > 1,
\end{equation}
where $\mathfrak{S}$ is the set of all states on $\mathcal{A}$.
\end{theorem}
\begin{proof}
Let $\rho = \sum_k \lambda_k \rho_k$ be the decomposition into pure states $\rho_k$ (i.e. $\rho_k^2 = \rho_k,\ \forall k $). Denoting
$$
\varphi_k (A) = \mathrm{Tr} \rho_k A\ (\forall A \in \mathcal{A}),
$$
then $\varphi = \sum \lambda_k \varphi_k$ is the extremal decomposition. 
Furthermore, if $\varphi \in {\rm ex}\mathfrak{S}$, $\rho$ is a pure state (i.e. $\rho = \rho^2$). Therefore according to Lemma \ref{lem1},
$$
 S_{\alpha} (\varphi) = \inf \{(1 - \alpha )^{-1} \log \sum \lambda_k^{\alpha} \} = S_{\alpha} (\rho)
$$
holds.
\end{proof}

\noindent
Therefore, if $\alpha > 1$, $\mathcal{S}$-mixing R\'enyi entropy includes the quantum R\'enyi entropy as the special case. If $0 \le \alpha < 1 $, the following inequality holds.
%Th 0\le \alpha < 1
\begin{theorem}
Under the above settings, for any $0\le \alpha<1$,
\begin{equation}\label{th01}
S_{\alpha}^{\mathfrak{S}} (\varphi) \le S_{\alpha}(\rho).
\end{equation}
\end{theorem}
%Pf
\begin{proof}
If $0\le \alpha < 1$, there holds
$$
(1 - \alpha)^{-1} \log \sum_{n} \lambda_n^{\alpha} \le (1 - \alpha)^{-1} \log \sum_{n} p_n^{\alpha}.
$$
This result induces the inequality (\ref{th01}).
\end{proof}

%Density Case----------

\subsection{Density Case}
Since $\mathcal{S}$-mixing R\'enyi entropy depends on $\mathcal{S}$, we can consider the complexity of the state measured from the reference system $\mathcal{S}$. In this chapter, we study the complexities of density operators by taking different reference systems.\\

\noindent Let $\mathbf{C}(\mathcal{H})$ be the set of all compact operators on $\mathcal{H}$. Then $\mathcal{A} := \mathbf{C} (\mathcal{H}) + \mathbb{C}I$ is a $C^*$-algebra. Now let $\theta (\mathbb{R})$ be the set of all 1-parameter strongly continuous automorphisms on $\mathcal{A}$ and let
$$
\theta_t (\cdot) := U_t \cdot U_{-t} \quad, \quad \theta_t \in \theta (\mathbb{R})
$$
where $U_t$ is a unitary operator on $\mathcal{A}$.\\
Furthermore, when $\mathcal{S} = \mathfrak{S}$, we simply denote  $S_{\alpha}^{\mathfrak{S}} (\varphi)$ by $S_{\alpha} (\varphi)$. 

%Th of Invariant
\begin{theorem}
If $\varphi$ is faithful and $\theta$-invariant, and if eigenvalues of $\rho$ are non-degenerate,
\begin{equation}
S_{\alpha}^{I (\theta)} (\varphi) = S_{\alpha} (\varphi)
\end{equation}
holds.
\end{theorem}
%proof
\begin{proof}
Since $\varphi \in I (\theta)$, for any $t \in \mathbb{R}$ and unitaries $U_t$, $[U_t, \rho ] = 0$ holds. Moreover, if $\varphi$ is faithful, $\rho > 0$ is satisfied. Furthermore, since the eigenvalues of $\rho$ are non-degenerate, we can put $\rho = | x_k \> \< x_k |$ where $x_k$ are any eigenvectors of $\rho$. \\
\noindent Therefore, for any $t \in \mathbb{R}$ and any $k$, 
$$
[U_t , \rho_k ] = 0
$$
holds. Hence $\rho_k \in I (\theta)$. Thus we obtain the following inequality:
$$
S_{\alpha} (\varphi) \geq S_{\alpha}^{I (\theta)} (\varphi).
$$
Next, we show the opposite inequality. Let $\varphi = \sum \lambda_k \varphi_k$ be the ergodic decomposition (i.e. $\varphi_k \in {\rm ex} I (\theta)$), and $\rho_k$ be a density adjusted $\varphi_k$. Then $\rho_k$ is a pure state. Therefore $\varphi_k \in {\rm ex} \mathfrak{S}$. Hence
$$
S_{\alpha} (\varphi) \le S_{\alpha}^{I (\theta)} (\varphi).
$$ 
\end{proof}

%Th of KMS
\begin{theorem}
If $\varphi \in K_{\beta} (\theta)$, $S_{\alpha}^{K (\theta)} = 0$.
\end{theorem}
%proof
\begin{proof}
Let $H$ be a Hamiltonian of a physical system, and $\beta := 1/kT$ ($k$ ; the Boltzmann constant, $T$ ; the temperature). Denote
$$
\rho = \frac{e^{-\beta H}}{{\rm Tr} e^{-\beta H}} \quad ,\quad e^{-\beta H} \in \mathbf{T} (\mathcal{H})
$$
and
$$
\varphi (A) := \mathrm{Tr}\rho A \quad, \quad A \in \mathcal{A}.
$$
Then $\varphi$ is a unique KMS state for $\beta$ and $\theta_t (A) := u_t A u_{-t}$ ($u_t := {\rm exp} (itH)$). Therefore, if $\varphi \in K_{\beta} (\theta)$, from uniqueness of a state,
$$
S_{\alpha}^{K (\theta)} (\varphi) = 0.
$$
\end{proof}

%General Case----------
\subsection{General Case}

In this section, we study the complexities of general states by taking different $\mathcal{S}$. 

%Th of inequalities for K, I, S
\begin{theorem}\label{KMS}
For any KMS states $\varphi \in K_{\beta} (\theta)$, the following inequalities hold:
\begin{enumerate}
\item $S_{\alpha}^{I (\theta)} (\varphi) \geq S_{\alpha}^{K (\theta)} (\varphi)$.
\item $S_{\alpha} (\varphi) \geq S_{\alpha}^{K (\theta)} (\varphi)$.
\end{enumerate}
\end{theorem}
%proof
\begin{proof}
1.\ The decomposition from $\varphi \in K_{\beta} (\theta)$ into ${\rm ex} K_{\beta}(\theta)$ is unique \cite{br1}. We put the decomposition $\varphi = \sum \lambda_n \varphi_n$. Then $\varphi_n \perp \varphi_m \ (n \neq m)$ holds.\
On the other hand, since ${\rm ex} K_{\beta} (\theta) \subset I (\theta)$, $\varphi_n$ can be decomposed into the elements of ${\rm ex}I (\theta)$, that is, ergodic states. Let $\varphi_n = \sum \mu_k^{(n)} \psi_k$ ($\psi_k \in {\rm ex}I (\theta)$) be the ergodic decomposition. Because of the uniqueness of the decomposition into $\varphi_n$,  we can regard $ (1-\alpha)^{-1} \log \sum_{n} (\lambda_n)^{\alpha}$ as the constant. Furthermore, $0 \le \sum_k (\mu_k^{(n)})^{\alpha} < 1$ holds. Therefore we have
\begin{eqnarray*}
(1-\alpha)^{-1} \log \sum_{k, n} (\lambda_n \mu_k^{(n)})^{\alpha} &=& (1-\alpha)^{-1} \log \sum_{n} (\lambda_n)^{\alpha} \sum_k (\mu_k^{(n)})^{\alpha} \\
&=& \frac{1}{\alpha - 1} \left\{ - \log \sum_n \lambda_n^{\alpha} + \left( - \log \sum_{n, k} (\mu_k^{(n)})^{\alpha} \right) \right\} \\
&\geq& (1 - \alpha)^{-1} \log \sum_n \lambda_n^{\alpha} =  S_{\alpha}^{K(\theta)}(\varphi).
\end{eqnarray*}
By taking the infimum over all $\{ \mu_k^{(n)}\}$, we obtain $S_{\alpha}^{I (\theta)} (\varphi) \geq S_{\alpha}^{K (\theta)} (\varphi)$.\\

\noindent 2. Since ${\rm ex} K_{\beta} (\theta) \subset \mathfrak{S}$, we obtain the inequality in the same way as 1.
\end{proof}

%G-commutative
\noindent Moreover, in order to investigate the inequality between $S_{\alpha}^{I (\theta)} (\varphi)$ and $S_{\alpha} (\varphi)$, we need $G$-commutativity of $(\mathcal{A}, \theta (G))$. Thus, we recall the definition.\\
\noindent Let $(\mathcal{H}_{\varphi}, \pi_{\varphi}, x_{\varphi})$ be the $GNS$-representation defined by $\varphi$ and $\{ u_g^{\varphi}\ ;\ g \in G \}$ be the strongly continuous unitary group on $\mathcal{H}_{\varphi}$.

\begin{definition}
Let $E_{\varphi}$ be a projection from $\mathcal{H}_{\varphi}$ to the set of $u_g^{\varphi}$-invariant vectors. If $E_{\varphi} \pi_{\varphi}(\mathcal{A})'' E_{\varphi}$ is a commutative von Neumann algebra, $(\mathcal{A}, \theta (G))$ is called {\it G-commutative} for $\varphi$.
\end{definition}

%Th of I(\theta)
\noindent Furthermore, we mention the following theorem.
\begin{theorem}\label{erg}
For $\varphi \in I (\theta)$, the following are satisfied:
\begin{enumerate}
\item There exists $\mu \in \mathcal{O}_{\varphi} (I(\theta))$ whose pseudo-support is ${\rm ex} I (\theta)$.
\item If $(\mathcal{A}, \theta (G))$ is $G$-commutative, $I (\theta)$ is a Choquet simplex. Therefore, then the above $\mu$ is a unique maximal measure. 
\end{enumerate}
\end{theorem}

%Th of inequalities
\noindent Now we prove the following inequalities.
\begin{theorem}
If $(\mathcal{A}, \theta (\mathbb{R}))$ is $G$-commutative for $\varphi$, 
\begin{equation}
S_{\alpha} (\varphi) \geq S_{\alpha}^{I(\theta)} (\varphi) \geq S_{\alpha}^{K(\theta)} (\varphi).
\end{equation}
\end{theorem}
%proof
\begin{proof}
According to Theorem \ref{erg},   the ergodic decomposition of $\varphi$ is unique. Hence the first inequality is satisfied. The second one holds from Theorem \ref{KMS}.
\end{proof}

\end{document}